\documentclass[letterpaper, 10 pt, conference]{ieeeconf}  

\usepackage{xurl}
\usepackage[hidelinks]{hyperref}

\usepackage{cite}
\usepackage{amsmath,amssymb,amsfonts,mathtools,amsthm,mathrsfs}
\usepackage{textcomp}
\usepackage{xcolor}
\usepackage{tikz,graphics,color,float,epsf}
\usepackage[ruled,vlined]{algorithm2e}

\DeclareMathOperator*{\argmax}{arg\,max}

\newtheorem{theorem}{Theorem}

\newtheorem{proposition}[theorem]{Proposition}

\SetKw{Assert}{Assert}
\SetKwInOut{Require}{Require}

\usepackage{graphicx}
\usepackage{subcaption}
\usepackage{balance}
\usepackage{eso-pic}

\IEEEoverridecommandlockouts                              

\title{\LARGE \bf
Computing Robustness to Sparse Diagonal Uncertainty
}

\author{Christoffer Kjellson and Emma Tegling
\thanks{The authors are with the Department of Automatic Control and the ELLIIT Strategic Research Area at Lund University. Email: \{{\tt\small{christoffer.kjellson, emma.tegling}\}@control.lth.se.}}%
\thanks{This work is partially funded by the Wallenberg AI, Autonomous Systems and Software Program (WASP) funded by the Knut and Alice Wallenberg Foundation.}
}

\begin{document}
\bstctlcite{IEEEexample:BSTcontrol}

\maketitle
\thispagestyle{empty}
\pagestyle{empty}

\AddToShipoutPictureFG*{%
  \AtPageLowerLeft{%
    \put(40,30){%
      \parbox{\dimexpr\paperwidth-80pt\relax}{%
        \centering\footnotesize
        This work has been submitted to the IEEE for possible publication.
        Copyright may be transferred without notice, after which this version
        may no longer be accessible.
      }%
    }%
  }%
}

\begin{abstract}
A new robustness metric~$\nu$ was recently proposed as a substitute for the structured singular value~$\mu$ to better capture robustness to sparse diagonal uncertainty, but its computation has remained an open problem. In this paper, we show that computing~$\nu$ is equivalent to maximizing the spectral radius of a nonnegative matrix product. This equivalence allows us to transfer existing results on spectral-radius maximization to~$\nu$, including a refined upper bound and conditions under which the bounds coincide. We then provide reformulations and structural results that enable an algorithm to solve the nonconvex optimization problem for nontrivial problems using global solvers. This also enables identification of the most fragile parts of the system. Our results do not yet provide a scalable solution for computing~$\nu$, but they are an important step toward computability and interpretability.
\end{abstract}

\section{INTRODUCTION}

In this paper, we study the robustness of linear time-invariant systems to diagonal structured uncertainty, as visualized by the interconnection in Fig.~\ref{fig:conn}. A common measure of robustness for such systems is the structured singular value from robust control~\cite{ssv}. However, quantifying robustness to diagonal uncertainty using the structured singular value does not capture the sparsity of the uncertainty well~\cite{kjellqvist2022}.

In many networked systems, uncertainty is naturally distributed across many local components, such as nodes, communication links, or physical interconnections. This has motivated the study of robustness to diagonal structured uncertainty within networked and interconnected systems~\cite{karow, zelazo, pates, mukherjee}. However, in such systems, a sparse failure mode may be much more relevant than one that affects many nodes or edges. Also, because many critical large-scale systems can be modeled as graphs, choosing the appropriate robustness measure for diagonal uncertainty is essential in many applications. These observations were made in~\cite{kjellqvist2022} and led to the development of the robustness metric $\nu$\footnote{The metric $\nu$ is unrelated to the more well-known Vinnicombe's $\nu$-gap.}, building on ideas from~\cite{dahleh}. Related approaches for reducing the conservativeness of the structured singular value were explored in~\cite{seungil} and~\cite{bamieh}, though in different problem settings.

\usetikzlibrary{positioning,calc}
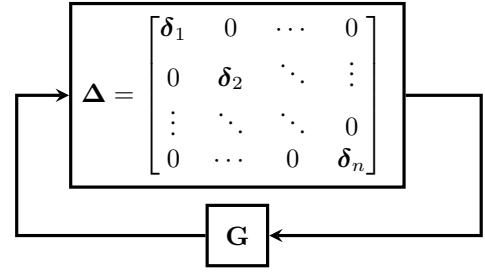
\begin{figure}
    \centering
    \begin{tikzpicture}[
        >=stealth,
        block/.style={
            draw,
            rectangle,
            minimum height=0.8cm,
            minimum width=0.8cm,
            align=center,
            inner sep=4pt,
            line width=1.2pt
        },
        node distance=0.2cm
    ]
    \node[block] (Delta) {
        $\boldsymbol{\Delta}=
        \begin{bmatrix}
            \boldsymbol{\delta}_1 & 0 & \cdots & 0 \\
            0 & \boldsymbol{\delta}_2 & \ddots & \vdots \\
            \vdots & \ddots & \ddots & 0 \\
            0 & \cdots & 0 & \boldsymbol{\delta}_n
        \end{bmatrix}$
    };
    \node[block, below=of Delta] (G) {$\mathbf{G}$};
    \draw[->, very thick] (Delta.east) -- ++(1.0,0) |- (G.east);
    \draw[->, very thick] (G.west) -- ++(-2.5,0) |- (Delta.west);
    \end{tikzpicture}
    \caption{A causal linear time-invariant operator $\mathbf{G}$ interconnected with diagonal structured uncertainty $\boldsymbol{\Delta}$.}\label{fig:conn}
\end{figure}

The authors of~\cite{kjellqvist2022} left the computation of $\nu$ unexplored, but the upper bound $\bar{\nu}$ has been used for controller synthesis using system-level synthesis (SLS)~\cite{li2022}. Even though the computation of the exact value of $\nu$ might not be useful for SLS, which requires a high degree of simplicity in the structure of the computation, being able to compute $\nu$ is still important. This would enable comparison of different systems and controllers with respect to $\nu$-robustness and provide a deeper understanding of this metric and its bounds.

Our first result is that computing $\nu$ is equivalent to maximizing the spectral radius of a certain nonnegative matrix product. Optimization of the spectral radius of matrices has been widely studied, for example, in~\cite{nesterov},~\cite{boyd2004convex}, and~\cite{axtell2009}, but the difficulty depends heavily on the specifics of the problem. We find that our exact formulation has been studied to a limited extent previously in~\cite{elsner2015} and~\cite{dasaratha2023}, motivated by resource allocation and growth maximization problems. We connect these results to $\nu$-analysis and further provide new reformulations and structural dimension-reduction results for the resulting optimization problems. Finally, we combine these results to propose an algorithm for computing $\nu$ for nontrivial examples. Importantly, this not only computes $\nu$, but also identifies the components of the system that are most fragile.

\section{PRELIMINARIES}

The robustness metric $\nu$ is defined as
\begin{equation*}
    \nu = \frac{1}{\inf\limits_{\boldsymbol{\Delta}\in\mathscr{D}}\{\sum_{i=1}^{n}\|\boldsymbol{\delta}_i\|_1 : (I-\mathbf{G}\boldsymbol{\Delta})^{-1}\text{ unstable}\}},
\end{equation*}
where, just as in~\cite{kjellqvist2022}, $\mathscr{D}$ is the set of diagonal $\ell_\infty$-stable causal linear time-varying operators, and $\mathbf{G}$ is an $n$-input, $n$-output stable causal linear time-invariant system. Since $\boldsymbol{\Delta}$ is diagonal, we let $\boldsymbol{\Delta}=\operatorname{diag}(\boldsymbol{\delta})$, where $\boldsymbol{\delta}_i$ denotes each diagonal element in $\boldsymbol{\Delta}$. The metric $\nu$ quantifies the sum of the largest magnitudes of each independent perturbation, and thus it accounts for the sparsity of the uncertainty.

For the same uncertainty set $\mathscr{D}$, the structured singular value is given by
\begin{equation*}
    \mu_{\mathscr{D}} = \frac{1}{\inf\limits_{\boldsymbol{\Delta}\in\mathscr{D}}\{\|\boldsymbol{\Delta}\|_{\infty,\infty} : (I-\mathbf{G}\boldsymbol{\Delta})^{-1}\text{ unstable}\}},
\end{equation*}
where $\|\cdot\|_{\infty,\infty}$ is the induced operator norm on $\ell_\infty$. Note that $\nu$ is only defined for diagonal uncertainty, which is why we exclude the subscript $\mathscr{D}$ in its definition, but keep it for $\mu$. We also take the dependence on $\mathbf{G}$ as implicit.

Now, let the magnitude matrix of $\mathbf{G}$ be defined as
\begin{equation*}
    M =
    \begin{bmatrix}
        \|G_{11}\|_1 & \hdots & \|G_{1n}\|_1 \\
        \vdots       & \ddots & \vdots \\
        \|G_{n1}\|_1 & \hdots & \|G_{nn}\|_1 \\
    \end{bmatrix}
\end{equation*}
where $G$ is the impulse response matrix of $\mathbf{G}$. By construction, $M$ is nonnegative, and each element $M_{ij}$ quantifies the worst-case gain between the uncertainty channels $i$ and $j$. Then, we can restate an important result from~\cite{kjellqvist2022}, showing that the computation of $\nu$ reduces to an optimization problem involving only the static matrix $M$.
\begin{theorem}[\!\cite{kjellqvist2022}, Proposition 1.1]
    Let $M$ be the magnitude matrix of $\mathbf{G}$. Then
    \begin{equation}\label{eq:nudef2}
        \nu=\frac{1}{\inf\limits_{\delta\in\mathbb{R}^n_+}\{\sum_{i=1}^{n}\delta_i : (I-\operatorname{diag}(\delta)M)\text{  }\mathrm{singular}\}}.
    \end{equation}
\end{theorem}

It will also be convenient to define the vector set $\mathcal{D}$ with implied dimension $m$, given by
\begin{equation}\label{eq:ddefs}
    \mathcal{D} = \{\gamma\in\mathbb{R}^m : \gamma\geq0, \mathbf{1}^\top\gamma=1\}.
\end{equation}
Here, and in the remainder of the paper, inequalities involving vectors or matrices are understood to be elementwise.

\subsection{Illustrative example}

Consider the systems
\begin{equation*}
    \mathbf{P}_1:\begin{cases}
        z_1(t\!+\!1)\!=\!\delta_1z_1(t)\\
        z_2(t\!+\!1)\!=\!\delta_2z_2(t)\\
        \quad\vdots\\
        z_n(t\!+\!1)\!=\!\delta_nz_n(t)
    \end{cases}
    \!\!\!\!\!\!,
    \mathbf{P}_2:\begin{cases}
        z_1(t\!+\!1)\!=\!\delta_1z_2(t)\\
        z_2(t\!+\!1)\!=\!\delta_2z_3(t)\\
        \quad\vdots\\
        z_n(t\!+\!1)\!=\!\delta_nz_1(t)
    \end{cases},
\end{equation*}
where  $\mathbf{P}_1$ is a decoupled first-order system and $\mathbf{P}_2$ is a delayed ring, both with uncertain system parameters~\cite{kjellqvist2022}. Destabilizing $\mathbf{P}_2$ minimally requires perturbing all states, while for $\mathbf{P}_1$, perturbing only one state with the same magnitude is enough. This means that $\mu_{\mathscr{D}}(\mathbf{P}_1)=\mu_{\mathscr{D}}(\mathbf{P}_2)=1$, but $\nu(\mathbf{P}_1)=1$ and $\nu(\mathbf{P}_2)=1/n$~\cite{kjellqvist2022}. This shows that $\nu$ better captures the amount of uncertainty needed to destabilize the system. However, for $\mathbf{P}_2$ the expression for~$\nu$ is easy to compute, but for the nontrivial systems which we will consider in Section~\ref{sec:dynsim}, the computation of~$\nu$ is no longer straightforward.

\section{EQUIVALENCE TO SPECTRAL RADIUS MAXIMIZATION}

We start by showing the connection between computing $\nu$ and spectral radius maximization of a matrix product in the following proposition.
\begin{proposition}\label{thm:matmul}
    Let $M$ be the magnitude matrix of $\mathbf{G}$, $\rho$ the spectral radius, and $\mathcal{D}$ as defined in~\eqref{eq:ddefs}. Then
    \begin{equation}\label{eq:equiv}
        \nu = \max_{\gamma\in\mathcal{D}}\rho(\operatorname{diag}(\gamma)M).
    \end{equation}
    Moreover, if $\nu>0$, then $\gamma^*$ is a maximizer of~\eqref{eq:equiv} if and only if
    \begin{equation*}
        \delta^*=\frac{1}{\nu}\gamma^*
    \end{equation*}
    is a minimizer of~\eqref{eq:nudef2}.
\end{proposition}
\begin{proof}
    Consider~\eqref{eq:nudef2}. If $\delta=\mathbf{0}$, then $(I-\operatorname{diag}(\delta)M)$ is not singular, so assume that $\delta\neq\mathbf{0}$. Then any permissible $\delta$ can be written as
    \begin{equation*}
        \delta=q\gamma,\qquad\gamma\in\mathcal{D},\qquad
        q=\mathbf{1}^\top\delta>0.
    \end{equation*}
    For a fixed $\gamma$, the smallest $q$ for which the matrix $(I-q\operatorname{diag}(\gamma)M)$ is singular is $q=1/\rho(\operatorname{diag}(\gamma)M)$ by the Perron--Frobenius theorem (we use the convention $1/0=\infty$). Since $\mathbf{1}^\top\delta=q$, minimizing the objective in~\eqref{eq:nudef2} is therefore equivalent to minimizing $1/\rho(\operatorname{diag}(\gamma)M)$ over $\gamma\in\mathcal{D}$. Hence,~\eqref{eq:equiv} holds, where the maximum exists because $\mathcal{D}$ is compact and the spectral radius is continuous.

    Now, if $\nu>0$, $\gamma^*$ is a maximizer of~\eqref{eq:equiv} and $\delta^*=\gamma^*/\nu$, it holds by construction that $\delta^*\in\mathbb{R}^n_+$ and $(I-\operatorname{diag}(\delta^*)M)$ singular, because $\rho(\operatorname{diag}(\delta^*)M)=\rho(\operatorname{diag}(\gamma^*)M)/\nu=1$. Hence, $\delta^*$ is a feasible solution to~\eqref{eq:nudef2}, which by the first part of the proof has objective value $1/\nu$, so it is a minimizer. Conversely, assume that $\delta^*$ minimizes~\eqref{eq:nudef2} and define $\gamma^*=\nu\delta^*$. Since $\sum_{i=1}^n\delta^*_i=1/\nu$, $\gamma^*\in\mathcal{D}$. As argued in the first part, the smallest scaling $q$ for which $(I-q\operatorname{diag}(\gamma^*)M)$ is singular is $q=1/\rho(\operatorname{diag}(\gamma^*)M)=1/\nu$. Thus, $\gamma^*$ is a maximizer of~\eqref{eq:equiv}.
\end{proof}

The optimization problem in~\eqref{eq:equiv} has been studied previously in other contexts in~\cite{elsner2015} and~\cite{dasaratha2023}, motivated by resource allocation and growth maximization problems. This connection is natural, since finding the minimal perturbation magnitude $\delta^*$ with respect to the sum of each individual magnitude that destabilizes $\mathbf{G}$ can be interpreted as a resource allocation problem.

There are also useful bounds for $\nu$. In particular, an upper bound for $\nu$ can be obtained as the solution to the optimization problem
\begin{equation}\label{eq:nubar}
    \bar{\nu}=\inf\limits_{d\in\mathbb{R}^n_{++}}\max\limits_{i,j}\left\lbrace\frac{d_i}{d_j}M_{ij}\right\rbrace,
\end{equation}
having noted the equivalence between the upper bound definitions in~\cite{kjellqvist2022} and~\cite{elsner2015}.

Additionally,~\cite{elsner2015} provides a refined upper bound. We state the result for $\nu$, which follows directly from results in~\cite{elsner2015}. Here, we let $s(M)=\max_iM_{ii}$ be the largest diagonal value of $M$, which is a lower bound to $\nu$~\cite{kjellqvist2022}, and $\bar{\nu}^{\mathrm{ref}}$ denotes the refined upper bound.

\begin{theorem}\label{thm:rest1}
    \begin{equation*}
        \nu\leq\bar{\nu}^{\mathrm{ref}}\coloneq\Big(1-\frac{1}{n}\Big)\bar{\nu} + \frac{1}{n}s(M)\leq\bar{\nu}.
    \end{equation*}
\end{theorem}

Next, we also obtain an important equivalence from~\cite{elsner2015}, showing that the lower bound equals the upper bound if and only if the upper bound is exact.
\begin{theorem}\label{thm:rest2}
    \begin{equation*}
        \nu=\bar{\nu}\quad\iff\quad\bar{\nu}=s(M)
    \end{equation*}
\end{theorem}
From Theorem~\ref{thm:rest2}, we additionally conclude that for the cases when the infimum in~\eqref{eq:nubar} is attained by some $d$, Conjecture 3 in~\cite{kjellqvist2022} is confirmed.

We now use Proposition~\ref{thm:matmul} to provide reformulations and structural results that aid in computing~$\nu$.

\section{STRUCTURAL RESULTS AND COMPUTATION}

As indicated by the previous section, further computation of $\nu$ is only required when the upper and lower bounds are distinct, i.e., when $\bar{\nu}\neq s(M)$. In this case, we distinguish between $M>0$ and $M\geq0$, because if $M>0$, the problem simplifies to a problem of minimization of a sum of ratios, which also follows directly from results in~\cite{elsner2015}.
\begin{theorem}\label{thm:rest3}
    Let $M$ be the magnitude matrix of $\mathbf{G}$, $\mathcal{D}$ as defined in~\eqref{eq:ddefs}, and assume that $M>0$. Then,
    \begin{equation}\label{eq:minim}
        \nu = \Big(\min\limits_{y\in\mathcal{D}}\sum_{i=1}^n\frac{y_i}{(My)_i}\Big)^{-1}.
    \end{equation}
\end{theorem}
Equation~\eqref{eq:minim} cannot in general be used for nonnegative matrices, because then $(My)_i$ can be zero on $\mathcal{D}$.

\subsection{$M>0$}

The case where the magnitude matrix has all nonzero elements, i.e., $M>0$, typically occurs when every input affects every output in $\mathbf{G}$. The gradient descent approach described in~\cite{elsner2015} provides a method for computing a lower bound for $\nu$ for this case. Although global convergence from arbitrary initializations is not guaranteed,~\cite{elsner2015} reports that their numerical experiments did not encounter nonglobal stationary points and suggests that stationary points other than the global extrema should be saddle points. We show that local minima that are not global can in fact occur.

Consider the $3\times3$ example
\begin{equation}\label{eq:example}
    M=
    \begin{bmatrix}
        4 & 16 & 12 \\
        8 & 4 & 4 \\
        1 & 16 & 8 \\
    \end{bmatrix}.
\end{equation}
Fig.~\ref{fig:contour} shows a contour plot of the objective function in~\eqref{eq:minim}, given by $f_M(y)=\sum_{i=1}^ny_i/(My)_i$. The contour plot shows two local minima with objective values of $0.131$ and $0.125$, but only the minimum near $y=[0,0,1]$ is global.
\begin{figure}
    \centering
    \vspace{1mm}
    \includegraphics[width=.9\linewidth]{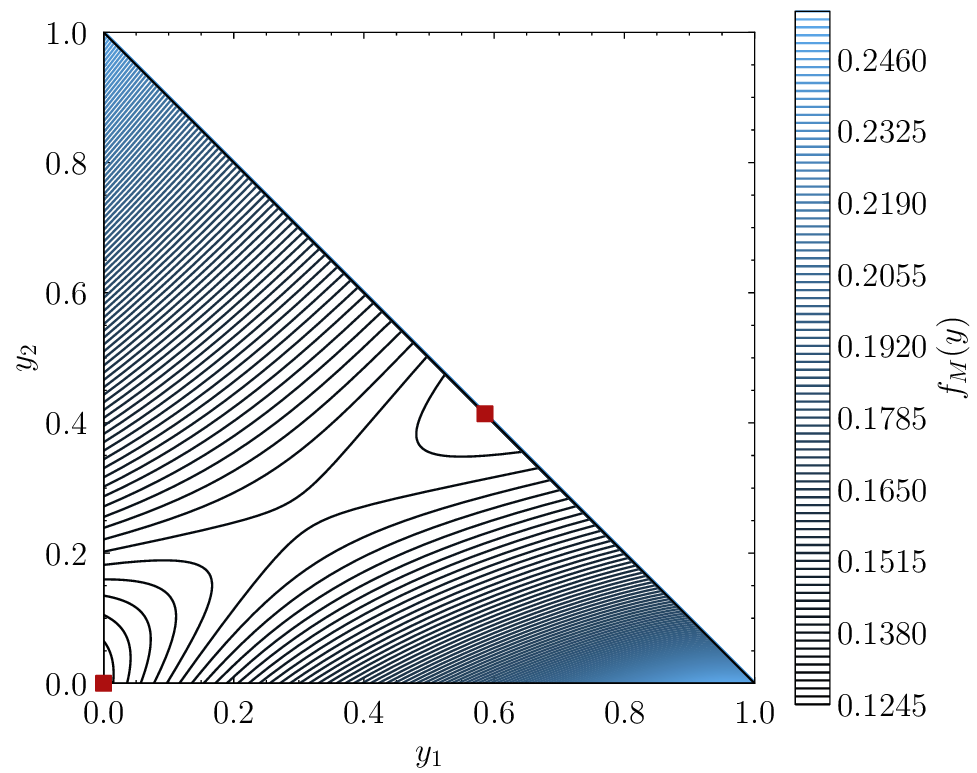}  
    \caption{Contour plot of the objective function in~\eqref{eq:minim} for example~\eqref{eq:example}, and the two local minima with distinct objective values, implying that gradient descent is not guaranteed to find the global optimum from an arbitrary initialization. Note that $y_3=1-y_1-y_2$ due to the constraints from $y\in\mathcal{D}$.}
    \label{fig:contour}
\end{figure}
Gradient descent can therefore not certify global solutions from arbitrary initializations. However, as we will discuss next,~\eqref{eq:minim} can still, in some cases, be solved globally.

Consider~\eqref{eq:minim}. By defining $\delta_i=y_i/q_i$, $q_i=(My)_i$, this can be reformulated into what is commonly called a bilinear program, or a quadratically constrained program (QCP), resulting in the minimization problem
\begin{equation}\label{eq:bilin}
    \begin{aligned}
        \nu^{-1} = & \min_{\delta,q,y}\sum_{i=1}^{n}\delta_i \\
        \text{s.t.}\quad & \mathbf{1}^\top y=1, & q_i = \sum_{j=1}^{n}M_{ij}y_j, \\
         & \delta_iq_i=y_i, & y\geq0.
    \end{aligned}
\end{equation}
In some cases, this minimization problem can be solved globally using relaxations with branch-and-bound methods~\cite{horst1996}, which are implemented in software packages such as SCIP~\cite{scip}, used in later sections. By a global solution, we mean that the solver's optimality gap is below a given tolerance. Such approaches are not expected to scale to large $n$ since the number of nodes required in branch-and-bound can grow exponentially with the number of dimensions.

Before moving on to the case $M\geq0$, we provide a set of three conditions sufficient for decreasing the dimension of the problem by identifying input-output channels that are irrelevant for our robustness analysis, and can thereby be excluded.
\begin{proposition}\label{thm:pruning}
    Assume that $M>0$, and let the index $t\in\{1,\ldots,n\}$ and $e_t$ denote the t-th standard basis vector. Suppose that there exists a $p\in\mathcal{D}$ with $\mathcal{D}$ in~\eqref{eq:ddefs}, $p_t<1$ and a vector $\beta\in\mathbb{R}^n_{++}$ such that
    \begin{equation}\label{eq:domcond}
        \begin{aligned}
            i)\text{  }& Mp\geq Me_t,\\
            ii)\text{  }& M_{ij}\geq\beta_iM_{tj}
            \quad\forall i,j,\\
            iii)\text{  }& \sum_{i=1}^n\frac{p_i}{\beta_i}\leq1.
        \end{aligned}
    \end{equation}
    Then there exists a minimizer $y^*$ to~\eqref{eq:minim} satisfying $y_t^*=0$.
\end{proposition}
\begin{proof}
    Assume that $y^*$ is a minimizer of~\eqref{eq:minim}, and suppose that $y^*_t>0$. Then, consider the modified solution $\hat{y}^*=y^* + \omega(p-e_t)$, where $\omega=y^*_t/(1-p_t)$. This enforces $\hat{y}^*_t=0$ and $\hat{y}^*\in\mathcal{D}$. By condition~$i)$,
    \begin{equation*}
        M\hat{y}^* = My^* + \omega(Mp-Me_t) \geq My^*.
    \end{equation*}
    Hence, using this elementwise inequality, we get
    \begin{equation*}
        \begin{aligned}
            \sum_{i=1}^n \frac{\hat{y}^*_i}{(M\hat{y}^*)_i} & \leq \sum_{i=1}^n \frac{\hat{y}^*_i}{(My^*)_i}\\
            & = f_{M}(y^*) + \omega\biggl(\sum_{i=1}^n \frac{p_i}{(My^*)_i} - \frac{1}{(My^*)_t}\biggr),
        \end{aligned}
    \end{equation*}
    where $f_M$ again denotes the objective function in~\eqref{eq:minim}. Condition~$ii)$ implies
    \begin{equation*}
        (My^*)_i\geq\beta_i(My^*)_t,
    \end{equation*}
    and therefore, using condition~$iii)$,
    \begin{equation*}
        \sum_{i=1}^n\frac{p_i}{(My^*)_i}\leq\frac{1}{(My^*)_t}\sum_{i=1}^n\frac{p_i}{\beta_i}\leq\frac{1}{(My^*)_t}.
    \end{equation*}
    Thus $f_M(\hat{y}^*)\leq f_M(y^*)$. Since $y^*$ is optimal,
    $\hat{y}^*$ is also optimal, and by construction the proposition follows.
\end{proof}
Proposition~\ref{thm:pruning} identifies coordinates whose mass can always be redistributed without increasing the objective, so an optimum exists with those coordinates removed. The vector $\beta$ can be computed directly from the second condition in~\eqref{eq:domcond}, which can then be used to solve for $p$ by a linear program. Proposition~\ref{thm:pruning} applies recursively, i.e., excluding some indices may enable further removals.

\subsection{$M\geq0$}

In the case of $M\geq0$ we can no longer guarantee that $(My)_i$ in~\eqref{eq:minim} is nonzero for all permissible $y$. However, in many cases, we can still compute $\nu$ in a manner similar to the $M>0$ case. Since $\operatorname{diag}(\gamma)M$ in Proposition~\ref{thm:matmul} is nonnegative, according to Perron-Frobenius theory, there exists a nonnegative eigenvector $x$ such that $\operatorname{diag}(\gamma)Mx = \rho(\operatorname{diag}(\gamma)M)x$. Let $x$ be normalized such that $\mathbf{1}^\top x=1$. Then for every feasible $\gamma$ there exists a $\lambda\geq0$ and $x$ such that $\lambda=\rho(\operatorname{diag}(\gamma)M)$ and $\operatorname{diag}(\gamma)Mx = \lambda x$. This leads to the maximization problem
\begin{equation}\label{eq:bilin2}
    \begin{aligned}
        \nu = & \max_{\gamma,x,\lambda} \lambda \\
        \text{s.t.}\quad & \mathbf{1}^\top\gamma=\mathbf{1}^\top x=1, & x\geq0, \\
         & \gamma_i\sum_jM_{ij}x_j=\lambda x_i, & \gamma\geq0,
    \end{aligned}
\end{equation}
which is also a QCP, but differs from~\eqref{eq:bilin}. In our numerical experiments, the solver finds a globally certified solution in fewer examples for this problem, but it relaxes the requirement from $M>0$ to $M\geq0$.

For reducible $M$, its spectral radius is determined by its irreducible diagonal blocks in Frobenius normal form (see, e.g.,~\cite{berman}). However, it is not immediate that the maximization problem in~\eqref{eq:equiv} decomposes in the same way. In the following proposition, we show that it does, proving that some minimal destabilizing perturbation can always be realized on one strongly connected component of the graph of $M$.

\begin{proposition}\label{thm:irred}
    Let $M_{S}$ denote the principal submatrix obtained from extracting the rows and columns of $M$ with indices $S$, and let $M_{S_k}$, $k\in\{1,\ldots,r\}$, denote the diagonal blocks of a Frobenius normal form of the magnitude matrix $M$, with index sets $S_k$. Then
    \begin{equation}\label{eq:nusub}
        \nu = \max_{k}\left(\max_{\gamma\in\mathcal{D}} \rho(\operatorname{diag}(\gamma)M_{S_k})\right).
    \end{equation}
\end{proposition}
\begin{proof}
    Let $P$ be a permutation matrix such that $PMP^\top$ is a Frobenius normal form of $M$. Letting $\Gamma=\operatorname{diag}(\gamma)$, we get
    \begin{equation*}
        \rho(\Gamma M)=\rho(P\Gamma M P^\top)=\rho((P\Gamma P^\top)(PMP^\top)).
    \end{equation*}
    Since $P\Gamma P^\top$ is diagonal with the entries of $\gamma$ permuted, and $\mathcal{D}$ is invariant under permutations, maximizing over $\gamma\in\mathcal{D}$ gives the same value for $M$ and $PMP^\top$. Hence, we may assume without loss of generality that $M$ is in Frobenius normal form. We then have
    \begin{equation*}
        \Gamma M = \begin{bmatrix}
            \Gamma_{S_1}M_{S_1} & \cdots & * \\
            \vdots & \ddots & \vdots \\
            0 & \cdots & \Gamma_{S_r}M_{S_r}
        \end{bmatrix},
    \end{equation*}
    so $\rho(\Gamma M)=\max_k\rho(\Gamma_{S_k}M_{S_k})$ because the matrix is upper block triangular. Now, let $\alpha_k=\sum_i(\Gamma_{S_k})_{ii}$, i.e., the mass distributed on each block. If $\alpha_k=0$, then $\Gamma_{S_k}=0$ and the corresponding spectral radius is zero. If $\alpha_k>0$, define $\hat\gamma_k=\gamma_{S_k}/\alpha_k\in\mathcal{D}$ and $\hat{\Gamma}_k=\operatorname{diag}(\hat{\gamma}_k)$. Then,
    \begin{equation*}
        \begin{aligned}
            \rho(\Gamma_{S_k}M_{S_k}) &=\alpha_k\rho(\hat{\Gamma}_kM_{S_k})\\
            &\leq\max_{\gamma\in\mathcal{D}}\rho(\operatorname{diag}(\gamma)M_{S_k}),
        \end{aligned}
    \end{equation*}
    since $\alpha_k\leq1$. This gives
    \begin{equation}\label{eq:maxineq}
        \max_{\gamma\in\mathcal{D}}\rho(\operatorname{diag}(\gamma)M)\leq\max_k\Big(\max_{\gamma\in\mathcal{D}}\rho(\operatorname{diag}(\gamma)M_{S_k})\Big).
    \end{equation}
    Now, let $k^*$ optimize the right-hand side in~\eqref{eq:maxineq}, and let $\gamma^*$ be optimal for that block. Construct $\gamma$ by setting $\gamma_{S_{k^*}}=\gamma^*$ and $\gamma_i=0$ outside of $S_{k^*}$. Then
    \begin{equation*}
        \rho(\operatorname{diag}(\gamma)M)=\rho(\operatorname{diag}(\gamma^*)M_{S_{k^*}}),
    \end{equation*}
    which proves~\eqref{eq:nusub} since the inequality in~\eqref{eq:maxineq} is always attained.
\end{proof}

The irreducible diagonal blocks of the Frobenius normal form of $M$ can be calculated as the strongly connected components of the directed graph of $M$ using standard algorithms.

In this section, we have proposed two reformulations of the optimization problem for computing $\nu$, and two structural results that can reduce the dimension of these optimization problems. We now put these ideas together and verify their usefulness in the following algorithm.

\subsection{Algorithm for computing $\nu$}

An algorithm\footnote{The algorithm and code to reproduce plots and results are available at: \url{https://github.com/ckjellson/ComputingNu}.} to compute $\nu$ together with the magnitude of one minimal destabilizing perturbation $\delta^*$ is presented in Algorithm~\ref{alg:nucomp}. The algorithm first uses the upper and lower bounds to filter out the easy cases. Next, since the spectral radius of a nonzero nonnegative irreducible matrix is strictly positive, $\nu_{S_k}=0$ only if $M_{S_k}=0$, so these cases can be disregarded. Then, as described in the previous section, we handle the two cases $M>0$ and $M\geq0$ separately, applying Proposition~\ref{thm:pruning} if $M>0$. If the used solver cannot certify global optimality for a subproblem, the algorithm fails.

The algorithm returns $\nu$ and the magnitude of one minimal destabilizing perturbation~$\delta^*$. The minimizer of~\eqref{eq:nudef2} is in general not unique, but the algorithm returns one such minimizer. The nonzero elements of this vector thereby identify one of the most fragile parts of the system in the sense of $\nu$-robustness, which can also be useful information for controller synthesis or system analysis.

\begin{algorithm}\caption{Compute $\nu$, $\delta^*$}\label{alg:nucomp}
    \DontPrintSemicolon
    \KwIn{$M\in\mathbb{R}^{n\times n}$, $M\geq0$}
    \KwOut{$\nu$, $\delta^*$, status}
    Compute $s(M)$, $\bar{\nu}$, and $i^*=\argmax_i M_{ii}$\;
    \If{$s(M)=\bar{\nu}$}{
        $\nu\gets s(M),\quad\text{status}\gets\text{optimal}$\;
        $\delta^*\gets e_{i^*}/\nu$ if $\nu>0$, otherwise $\delta^*\gets\varnothing$\;
        \Return
    }
    Compute $S_k$ as defined in Prop.~\ref{thm:irred}\;
    \ForEach{$k$ such that $M_{S_k}\neq0$}{
        \eIf{$M_{S_k}>0$}{
            Recursively compute $R_k\subseteq S_k$ by
            Prop.~\ref{thm:pruning}\;
            $\delta_{R_k}\gets\mathrm{solve}$~\eqref{eq:bilin}
            on $M_{R_k}$\;
            $\nu_{S_k}\gets(\mathbf{1}^\top\delta_{R_k})^{-1}$\;
            $\delta_{S_k}\gets0$, \quad
            $(\delta_{S_k})_{R_k}\gets\delta_{R_k}$\;
        }{
            $(\nu_{S_k},\gamma_{S_k}^*)\gets\mathrm{solve}$~\eqref{eq:bilin2}\;
            $\delta_{S_k}\gets\gamma_{S_k}^*/\nu_{S_k}$\;
        }
        \If{solution not globally certified}{
            $\text{status}\gets\text{failed}$\;
            \Return
        }
    }
    $k^*\gets\argmax_k\nu_{S_k}$\;
    $\nu\gets\nu_{S_{k^*}},\quad\text{status}\gets\text{optimal}$\;
    $\delta^*\gets0,\quad\delta^*_{S_{k^*}}\gets\delta_{S_{k^*}}$\;
\end{algorithm}

\section{NUMERICAL EXAMPLES}

\subsection{Random matrices}

\begin{figure}
    \centering
    \vspace{1mm}
    \includegraphics[width=0.98\linewidth]{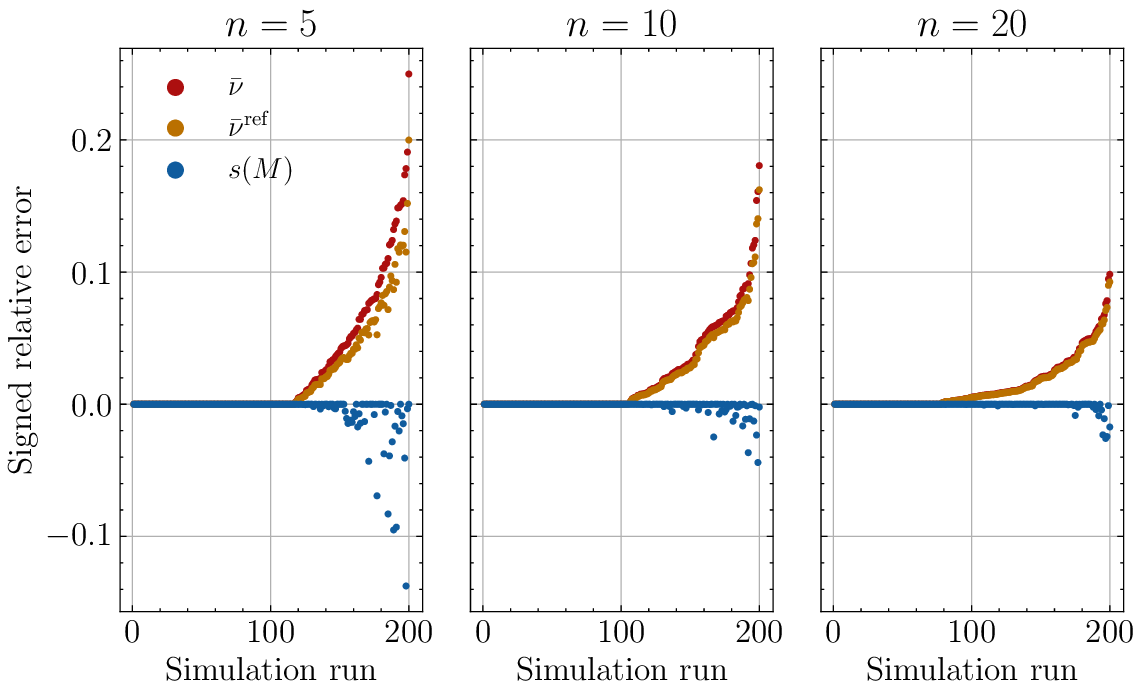}
    \caption{Signed relative errors of upper and lower bounds for 200 random matrices, sorted by increasing relative error. The lower bound is generally tighter than the upper bounds in our examples, and the bounds are often exact.}
    \label{fig:bounds}
\end{figure}

To demonstrate typical tightness of bounds and solvability, we first study random matrices $M=\mathbf{1}\mathbf{1}^\top+X$, where each element in the $n\times n$ matrix $X$ is sampled from a uniform distribution on the interval $[0,1]$ as in~\cite{elsner2015}. For each dimension $n\in\{5,10,20\}$, we generate 200 such matrices, compute all bounds, and also run Algorithm~\ref{alg:nucomp} to compute $\nu$ with SCIP~\cite{scip} as a global solver. Algorithm~\ref{alg:nucomp} succeeded in computing $\nu$ in all of the 600 cases.

For comparison of bounds, we consider the signed relative error of a bound $x$ given by $(x-\nu)/\nu$, where $\nu$ is computed by Algorithm~\ref{alg:nucomp}. The signed relative errors of $s(M)$, $\bar{\nu}$ and $\bar{\nu}^{\text{ref}}$ are plotted in Fig.~\ref{fig:bounds}, showing that the lower bound $s(M)$ is generally tighter than either upper bound in these examples. Also, $\bar{\nu}=s(M)$ in more than 50\% of the cases.

\subsection{Dynamical systems}\label{sec:dynsim}

\begin{figure*}[t]
    \centering
    \begin{subfigure}[t]{0.45\textwidth}
        \centering
        \includegraphics[width=\linewidth]{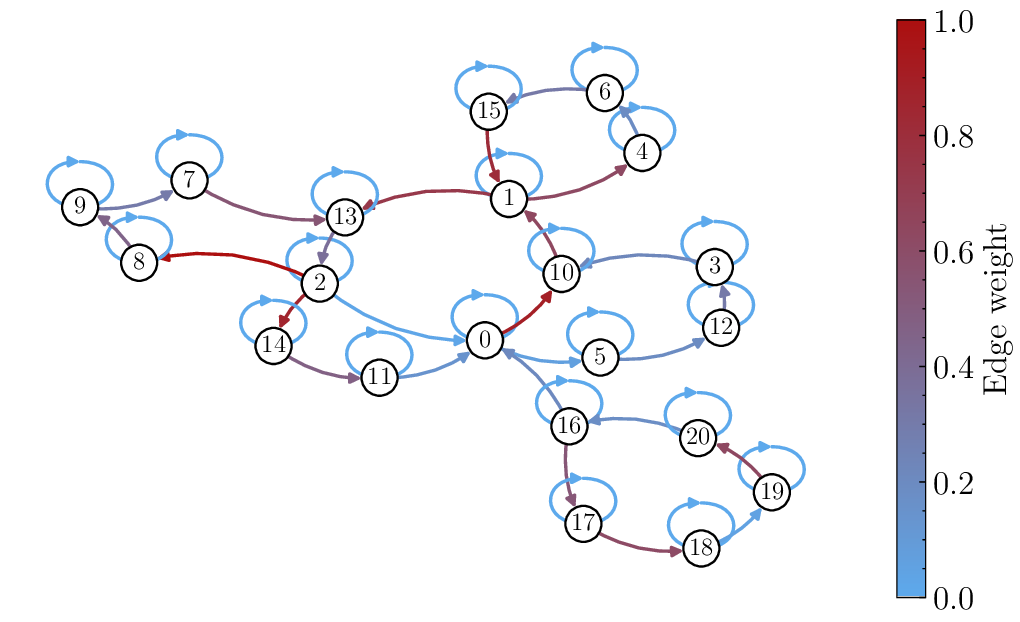}  
        \caption{}
        \label{fig:agraph}
    \end{subfigure}
    \hfill
    \begin{subfigure}[t]{0.45\textwidth}
        \centering
        \includegraphics[width=\linewidth]{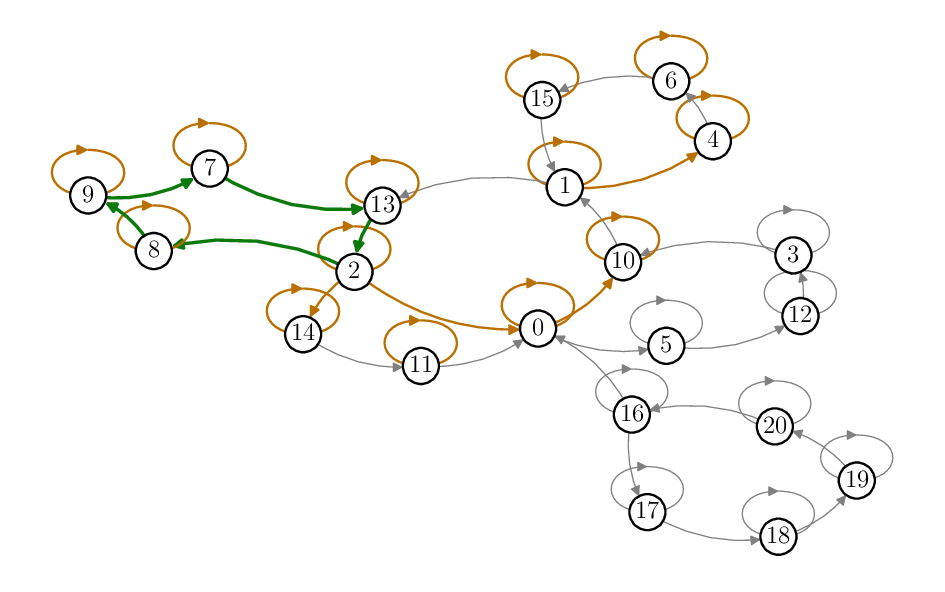}  
        \caption{}
        \label{fig:nugraph2}
    \end{subfigure}
    \caption{(a) Graph of $W$ in the dynamical system examples colored by edge weight. (b) Edges corresponding to the nonzero entries in~$\delta^*$ for System~2 in green, and the edges corresponding to indices reduced by Proposition~\ref{thm:pruning} in yellow. The minimal destabilizing perturbation is localized to the top left loop, indicating which part of the system is most fragile to uncertainty.}
    \label{fig:combined}
\end{figure*}

Although random matrices can say something about the general behavior of $\nu$ and its bounds, they do not necessarily represent what happens for magnitude matrices generated by a certain type of system $\mathbf{G}$. Such matrices are not random, but have an inherent structure that depends on $\mathbf{G}$, and this is not investigated in~\cite{kjellqvist2022}.

To illustrate this, we construct a weighted digraph with adjacency matrix $W$, which is visualized in Fig.~\ref{fig:agraph}. We then factorize this matrix as $BC=W$, such that when multiplying by $\mathbf{\Delta}\in\mathscr{D}$ according to $B\mathbf{\Delta}C$, each diagonal element in $\mathbf{\Delta}$ corresponds to a nonzero element in $W$. We consider two LTI systems with system matrices $A=\mathbf{0}$ and $A=W$, and with input and output matrices $B$ and $C$. Thus, closing the interconnection in Fig.~\ref{fig:conn}, we get the dynamics
\begin{equation}\label{eq:systems}
    \begin{aligned}
        \text{System 1:}&\quad z(t+1)=Bw(t),\quad\text{and}\\
        \text{System 2:}&\quad z(t+1)=Wz(t)+Bw(t),\\
        \text{with}&\quad v(t)=Cz(t),\quad w(t)=(\mathbf{\Delta}v)(t).
    \end{aligned}
\end{equation}

\begin{table}
    \caption{Value of $\nu$ and its bounds for the two example systems in~\eqref{eq:systems}. All bounds are loose.}
    \label{tab:results}
    \centering
    \begin{tabular}{l|l|l|l|l}\hline
         & $s(M)$ & $\nu$ & $\bar{\nu}^{\text{ref}}$ & $\bar{\nu}$ \\ \hline
        System 1 & 0.001 & 0.100 & 0.484 & 0.495 \\
        System 2 & 0.035 & 0.181 & 0.496 & 0.506 \\ \hline
    \end{tabular}
\end{table}

\begin{table}
    \caption{Runtimes for the computation of $\nu$ for System 2 from the ablation study.}
    \label{tab:times}
    \centering
    \begin{tabular}{l|l|l|l|l}\hline
         & Algorithm~\ref{alg:nucomp} & w/o P.\ref{thm:pruning} & w/o P.\ref{thm:irred} & w/o QCP~\eqref{eq:bilin} \\ \hline
        Runtime & 13s & 25s & timeout & timeout \\ \hline
    \end{tabular}
\end{table}

Using Algorithm~\ref{alg:nucomp} to compute $\nu$, and the convex formulation from~\cite{kjellqvist2022} to compute $\bar{\nu}$, we obtain the results in Table~\ref{tab:results}. This shows that both the lower and upper bounds are very loose, and therefore, computing $\nu$ exactly can provide substantially more informative results for such systems.

To illustrate how the computation of $\nu$ provides information about which part of the system is fragile as opposed to using $\mu$ or $\bar{\nu}$ as a robustness metric, we also show the edges corresponding to the nonzero elements in the resulting minimal destabilizing perturbation magnitude $\delta^*$ in Fig.~\ref{fig:nugraph2}. This also shows the edges corresponding to indices that are removed by Proposition~\ref{thm:pruning}, reducing the dimension of the largest subproblem from 36 to 22. Note that nodes 16-20 are not strongly connected to the rest of the graph, which, for System~2, results in two strictly positive subproblems.

Also, to illustrate the usefulness of the formulation in~\eqref{eq:bilin} and Propositions~\ref{thm:pruning} and~\ref{thm:irred}, we perform a simple ablation study with a timeout of 10 minutes. For details on how the exclusion of each component modifies Algorithm~\ref{alg:nucomp}, see the publicly available code. The resulting runtimes for computing $\nu$ for System 2 are reported in Table~\ref{tab:times}. This shows that all components of Algorithm~\ref{alg:nucomp} are useful for this example.

\section{CONCLUSION}

In this paper, we have established that computing the robustness metric $\nu$ proposed in~\cite{kjellqvist2022} is equivalent to maximizing the spectral radius of a matrix product, and have connected the metric to results in~\cite{elsner2015}. We then developed reformulations and structural results that enable computation of $\nu$ for nontrivial examples. Numerical examples illustrate both the usefulness of our formulations for computing $\nu$ and the computational gains from the dimension-reduction results. We emphasize that computing $\nu$ in this way not only provides a measure of robustness, but can also indicate which parts of the system are most vulnerable to uncertainty. Developing scalable computation and tighter bounds is an important direction for future work.


\section*{ACKNOWLEDGMENT}

The authors thank David Ohlin for directing our attention to this research direction and for helpful comments on the manuscript.

\bibliographystyle{IEEEtran}
\bibliography{bib}

\end{document}